\documentclass[11pt,onecolumn]{IEEEtran}
\IEEEoverridecommandlockouts

\usepackage{multirow}

\usepackage{enumitem}
\usepackage{framed}
\usepackage{amsmath}
\usepackage{amsfonts}
\usepackage{epsfig}
\usepackage{epsf}
\usepackage{subfigure}
\usepackage{graphicx}
\usepackage{cite}
\usepackage{url}
\usepackage{cite}
\usepackage[]{authblk}
\usepackage{array}
\usepackage{color}
\usepackage{colortbl}
\usepackage{amsthm}

\newtheorem{corollary}{Corollary}
\newtheorem{theorem}{Theorem}
\newtheorem{definition}{Definition}

\newtheorem{lemma}{Lemma}

\newtheorem{remark}{Remark}

\newcolumntype{K}{>{\centering\arraybackslash}m{3.5cm}}
\newcolumntype{L}{>{\centering\arraybackslash}m{2.5cm}}
\newcolumntype{M}{>{\centering\arraybackslash}m{1cm}}
\definecolor{Gray}{gray}{.9}

\title{Correcting Two Deletions and Insertions in Racetrack Memory}

\author{Alireza~Vahid,~Georgios~Mappouras,~Daniel~J.~Sorin,~Robert~Calderbank\\
				Department~of~Electrical~and~Computer~Engineering\\
				Duke~University
}

\begin{document}
\sloppy
\maketitle



\begin{abstract}
Racetrack memory is a non-volatile memory engineered to provide both high density and low latency, that is subject to synchronization or shift errors. This paper describes a fast coding solution, in which ``delimiter bits'' assist in identifying the type of shift error, and easily implementable graph-based codes are used to correct the error, once identified. A code that is able to detect and correct double shift errors is described in detail. 
\end{abstract}

\begin{IEEEkeywords}
Racetrack memory, synchronization error, shift error,  insertion/deletion channel, shift error correcting codes.
\end{IEEEkeywords}

\section{Introduction}
\label{Section:Introduction}

Historically, improvements in memory and storage have been due to advances in the memory technologies themselves together with innovations by computer architects who design memory and storage systems, and by coding theorists who design codes for storing data.

Racetrack memory is a non-volatile memory engineered to provide both high density and low latency, and aims to replace conventional memories such as DRAM and Flash. Racetrack memory stores data in ``tape-like'' tracks, and to read a stored bit, an electric current is injected to place the desired bit under the read-write port. Tracks can be manufactured in two or three dimensions, but the U-shaped three-dimensional geometry shown in Fig.~\ref{Fig:Racetrack3D} is preferred because it significantly increases recording density. Prior work on error resilience assumed multiple read/write ports on a single track, but this geometry requires a single read/write port. Competing technologies, such as phase-change memory (PCM) and magneto-resistive random-access memory (MRAM) cannot match the density of racetrack memory. In terms of latency, only SRAM (which is a volatile technology) has a slight advantage, but it comes at the cost of lower density~\cite{mittalsurvey,thomas2011racetrack,sun2013cross}. 

Racetrack uses current injection to {\it shift} magnetic domains that store data bits on the tracks. In this technology, a deletion occurs when the injected current is larger than expected causing one or more domains to be skipped. A repetition occurs when the injected current is smaller than expected, so that the domain under the port does not change, and we read the same bit two or more times. We refer to deletion and repetition errors as shift errors. These types of synchronization error define a channel that has been studied extensively by information theorists~\cite{gallager1961sequential,dobrushin1967shannon,ullman1967capabilities,diggavi2001transmission,kavcic2004insertion,diggavi2007capacity,mitzenmacher2009survey}. There is no closed form expression for the capacity of this insertion/repetition and deletion channel, but standard linear codes are known to be far from optimal. For example, the rate penalty for correcting a single shift error is logarithmic in the block length, but the optimal linear code has rate $1/2$~\cite{abdel2010correcting}.

There is some prior work in the computer architecture literature that seeks to handle shift errors by introducing extra read/write ports. The HiFi scheme presented in~\cite{HiFi} requires at least two extra reads and one extra write to access a stored bit. When an error is detected, it is corrected by reversing the current injection and reading the memory again. The extra read/write operations and the need for additional hardware limit the attractiveness of this solution and diminish the appeal of racetrack memory.

We propose an alternative solution, based on codes that detect and correct shift errors. The codes are easy to implement, and the circuits for encoding and decoding consume little power. Our coding solution does not require additional hardware, and it does not introduce additional reads or writes to memory.

Section~\ref{Section:Practical} describes a product code that can detect and correct up to two shift errors every $m+3$ shift operations in racetrack memory. The product structure makes it possible to decouple error detection from error correction, so that in the common case of no shift error, the access to stored data is fast. We detect up to two shift errors by combining Varshamov-Tenengolts (VT) codes~\cite{VTCodes} with blocks of delimiter bits. Our product construction extends the utility of VT codes, which are limited to correcting single shift errors, and we also provide a low complexity construction for these codes. The inner code detects up to two shift errors within a single track, and is able to correct one. The outer code connects data stored on different tracks, and makes it possible to correct two shift errors.

\section{Racetrack Memory}
\label{Section:Racetrack}

In this section we introduce racetrack memory and our error model.

\subsection{Racetrack Background}

Racetrack memory stores data in $r$ parallel tape-like tracks. Each data bit is stored in a magnetic domain and neighboring domains are separated by a domain wall as depicted in Fig.~\ref{Fig:Racetrack}. All read/write ports and the physical substrate are fixed in position, and as current is passed through the track, the domains pass by magnetic read/write ports positioned near the wire. 
\begin{figure}[t]
\centering
\includegraphics[height = 3.5cm]{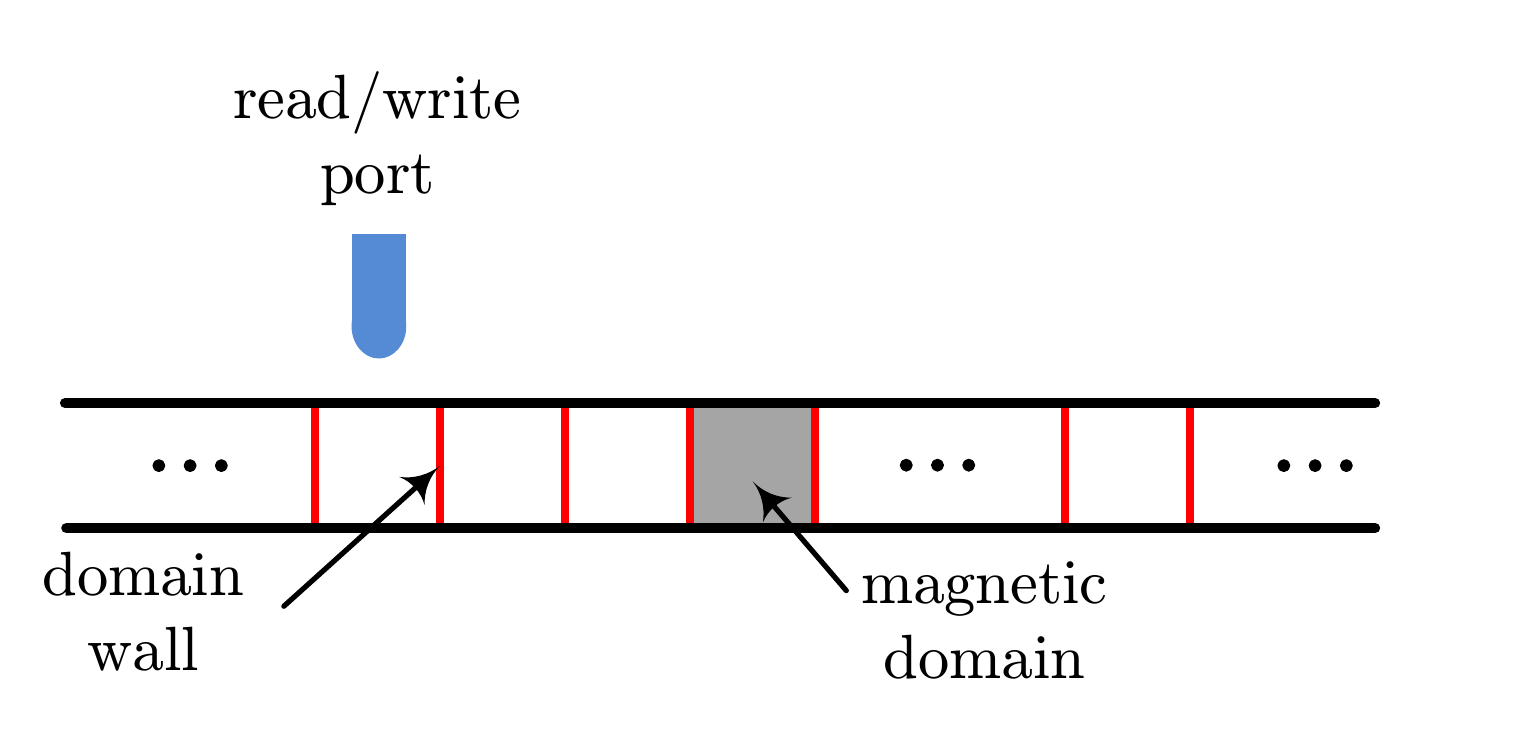}
\caption{A single track with one read/write port. One bit is stored per domain and two neighboring domains are separated by a domain wall.\label{Fig:Racetrack}}
\end{figure}

Racetrack memory suffers from deletion and repetition errors. Bits are stored on a track and are accessed by injecting a current to shift them and place them under the read/write port. A \emph{shift} is therefore the injection of the current in order to place the next bit under the read/write port. If the injected current is larger than expected, then we might skip one or multiple magnetic domains resulting in {\it deletion errors}. On the other hand, if the injected current is smaller than expected, the port's position might not change and we might read the same bit twice. We refer to this error as a {\it repetition error}. We refer to deletion and repetition errors as ``shift errors.'' 

In deletion channels as studied in information theory, the size of the output is equal to the size of the input minus the number of deletions. However, in racetrack memory, the memory controller is always going to provide the desired number of bits, say $n$, regardless of whether an error occurs, so that an error determines which $n$ bits are provided. This is an important characteristic of racetrack memory, and a consequence is that we first need to determine whether an error has occurred.

There is a difference between a {\it repetition} and an {\it insertion}. An insertion means introducing a new bit in a string regardless of the value of its preceding bits. However, a repetition requires the new bit to be equal to the bit right before it. Thus, a repetition error is an instance of an insertion error.

Our goal is to devise a practical code that would allow reliable data recovery in the presence of shift errors in racetrack memory, and for that we need to define the error model.

\subsection{Error Model}
\label{Section:ErrorModel}

We make the following assumptions:
\begin{enumerate}

\item In each track and in every $m+3$ shifts (shift is defined above), at most {\bf two} shift errors may occur. 

\item When reading $r$ parallel tracks, in every $m+3$ shifts (per track), at most one track may have two shift errors.

\end{enumerate}

As we see later, if the second assumption is violated, then we can detect the errors but we will only be able to correctly decode the data on the tracks that have a single shift error. A more restrictive model for deletion and insertion channels with segmented errors was used in~\cite{liu2010codes} in which only one shift error could happen per segment (a segment is a fixed number of input bits). Moreover in~\cite{liu2010codes}, authors assume either deletion errors or insertion errors but not a mixture of the two may occur in the channel; we impose no such constraint.

\section{Main Result}
\label{Section:Main}

In this section, we present our main contribution. First, we need the following definition.

\begin{definition}
If any $k$ input data bits can be mapped to some $m$-bit codeword such that the $k$ data bits can be recovered error-free under the assumptions of Section~\ref{Section:ErrorModel}, we say a zero-error rate of $k/m$ is achievable.
\end{definition}

\begin{theorem}
\label{THM:Main}
For racetrack memory with the error model defined in Section~\ref{Section:Racetrack} and with $r$ parallel tracks, we can achieve a zero-error rate of
\begin{align}
R_{0}\left( \ell, r \right) = \frac{r-1}{r} \times \frac{2^\ell - \ell - 1}{2^\ell+6}, \qquad \ell \in \mathbb{Z}^+,
\end{align}
where $k = 2^\ell - \ell - 1$ is the number of input data bits that are mapped to $m = 2^\ell+6$ magnetic domains on the memory.
\end{theorem}

Our code has three main components in two dimension. First dimension is within a single track where we use Varshamov-Tenengolts codes to store data. These codes are able to correct single shift errors. However, Varshamov-Tenengolts codes are not capable of detecting or correcting two shift errors. We enhance these codes by appending some helper bits that we refer to as {\it delimiter bits} which allow for double-shift-error detection per track. The role of the delimiter bits is to convert racetrack to a traditional deletion/insertion channel. The second dimension is across the $r$ tracks where we deploy a single parity-check code to correct double shift errors. The challenging part is to detect two shift errors and correct one. We use a simple outer code to correct two shift errors. 



The rest of the paper is dedicated to the proof of Theorem~\ref{THM:Main}. We also note that $R_{0}\left( \ell, r \right) \rightarrow 1$ as $\ell,r \rightarrow \infty$. As an example, consider $\ell = 6$ (\emph{i.e.} $57$ input data bits) and $r = 8$. Then, we have $R_{0}\left( 6, 8 \right) = 0.7125$.

\section{Proof of Theorem~\ref{THM:Main}: Practical Codes for Racetrack Memory}
\label{Section:Practical}

In this section, we introduce our code for racetrack memory and we prove Theorem~\ref{THM:Main}. We start with a single track and we devise a code that can correct a single shift error and detect up to two shift errors every $m+3$ shifts. We then show how to correct two shift errors using an outer code that takes advantage of the spatial diversity in racetrack memory. 

\begin{table*}[t]
\begin{center}
\begin{tabular}{|M|M|M|M|M|M|L|K|}
\hline
\multicolumn{7}{|c|}{Observation} & {\multirow{2}{*}{Decision}} \\
\cline{1-7}

position $m-5$ & position $m-4$ & position $m-3$ & position $m-2$ & position $m-1$ & position $m$ & Checksum {\bf mod~} $n+1$ & \\ \hline
1      &  1    &  0     & 0      &  {\sf X} &  {\sf X}  & =~0                 &  no error                   \\ \hline \rowcolor{Gray}
1      &  1    &  0     & 0      &  {\sf X} &  {\sf X}  & $\neq$~0            &  1 del. and 1 rep.: {\bf erasure}  \\ \hline 
1      &  0    &  0     & 0      &  {\sf X}  &  {\sf X}   & {\sf X}           &  1 deletion error           \\ \hline
{\sf X} &  1    &  1     & 0      &  0     &  {\sf X}   & {\sf X}             &  1 repetition error         \\ \hline 
0      &  0    &  0     & 0      &  {\sf X}     &  {\sf X}   & {\sf X}        &  2 deletion errors: {\bf erasure}         \\ \hline
{\sf X}      &  {\sf X}    &  1     & 1      &  0     &  0   & {\sf X}        &  2 repetition errors: {\bf erasure}       \\ \hline
\end{tabular}
\vspace{2mm}
\caption{Resulting observations in positions $m-2$, $m-1$ and $m$ for all possible error combinations and ${\sf X}$ indicates an irrelevant entry.\label{Table:Syndrom}}
\end{center}
\vspace{-6mm}
\end{table*}




On a single track, we encode the data bits into an {\it extended codeword} of length $m$ which consists of a Varshamov-Tenengolts code (defined below) of length $n = m-6$ followed by six helper bits that we refer to as delimiter bits.

\begin{definition}
The Varshamov-Tenengolts code $VT\left( n \right)$ is the set of all binary strings $\left( c_1, c_2, \ldots, c_n \right)$ satisfying
\begin{align}
\label{eq:VTCondition}
\sum_{i=1}^n{ic_i} \overset{\mathbf{mod}~n+1}\equiv 0, 
\end{align}
where the sum is evaluated as an ordinary rational integer\footnote{Technically speaking, our definition is a specific instance of VT codes. In general, VT codes satisfy $\sum_{i=1}^n{ic_i} \overset{\mathbf{mod}~n+1}\equiv a$, and we used $a=0$.}.
\end{definition}

\begin{remark}
VT codes were introduced in~\cite{VTCodes} to correct errors on a Z-channel. It follows from (\ref{eq:VTCondition}) that these codes are nonlinear over the binary field.
\end{remark}

\noindent {\bf Encoding:} It is possible to construct $VT\left( n \right)$ codes for any $n$, but to maximize the speed of encoding and decoding we choose $n = 2^\ell$ and $k = n-\ell-1$, because when $n+1 = 2^\ell + 1$ the necessary modular arithmetic can be executed more efficiently (see~\cite{piestrak1994design}). The encoding algorithm is described below.
\vspace{-2mm}
\begin{framed}
\begin{enumerate}[leftmargin=0.3cm]
\item Start with a zero-vector $\mathbf{c}$ of length $n = 2^\ell$.

\item Set positions that are not powers of two to data bits (there are $k = n-\ell-1$ such positions).

\item Set $s$ to be the minimum value that needs to be added to the checksum $\sum_{i=1}^n{ic_i}$ to make it equal to $0$ modulo $n+1$.

\item Set the $\ell+1$ positions that are powers of two to the binary expansion (of length $\ell+1$) of $s$. Start from $c_1$ and set it to the least significant bit of the binary expansion of $s$, and move all the way to $c_n$ and set it to the most significant bit of the binary expansion of $s$.
\end{enumerate}
\end{framed}
\vspace{-2mm}
It is straightforward to verify that the resulting codeword satisfies (\ref{eq:VTCondition}). We know that in step $3$ we have $0 \leq s \leq n$. Since $n = 2^\ell$, the binary expansion of $s$ is at most $\ell+1$ bits long. 

After encoding the $k$ data bits into a codeword $\mathbf{c}$ of size $n$, we append six delimiter bits to create the \textit{extended codeword} of length $m = n + 6$ that we write to the memory. The six delimiter bits are fixed and equal to $\mathbf{110000}$.

In the remainder of this section, we show that the extended codeword introduced above can correct a single shift error and detect up to two shift errors every $m+3$ shifts. Then, we show how we can correct two shift errors. We present our argument in five steps. In steps~1 and~2, we show that VT codes are capable of correcting a single shift error. Step~3 shows how we can detect up to two shift errors and takes into account the fact that the memory controller is always going to provide the desired number of bits, regardless of whether an error occurs. Step~4 show how to correct a single shift error using our extended codeword. Finally, step~5 incorporates spatial diversity to correct two shift errors.

\vspace{1mm}
\noindent{{\bf Step~1: VT codes are single-deletion-error correcting codes.}} Consider two binary strings $\mathbf{x} = x_1,x_2,\ldots,x_{n}$ and $\mathbf{y} = y_1,y_2,\ldots,y_{n}$. Let $D_{-1}\left( \mathbf{x} \right)$ and $D_{-1}\left( \mathbf{y} \right)$ denote the set of all the binary strings of length $n-1$ that result from a single deletion in $\mathbf{x}$ and $\mathbf{y}$ respectively. If we wish to use these two strings to store different values in a single-deletion channel, we need
\begin{align}
\label{Eq:CodeRequirement}
D_{-1}\left( \mathbf{x} \right) \cap D_{-1}\left( \mathbf{y} \right) = \emptyset.
\end{align}

The following lemma shows that VT codes are single-deletion-error correcting codes.

\begin{lemma}[\cite{levenshtein1966binary,levenshtein1965binary}]
\label{Lemma:VTGraphBased}
Any two binary strings $\mathbf{x}, \mathbf{y} \in VT\left( n \right)$, satisfy (\ref{Eq:CodeRequirement}).
\end{lemma}

For completeness, we present the proof of Lemma~\ref{Lemma:VTGraphBased} in Appendix~\ref{Appendix:AlgoDescription}.

\vspace{1mm}
\noindent{{\bf Step~2:  VT codes are single-insertion-error correcting codes.}} Note that a single-insertion-error correcting code is a single-repetition-error correcting code as well. Similar to what we described above, let $D_{+1}\left( \mathbf{x} \right)$ denote the set of all the binary strings of length $n+1$ that result from a single insertion in $\mathbf{x}$. We have the following result that shows the Varshamov-Tenengolts codes can be used for single-insertion-error-correction as well as single-deletion-error-correction.
\begin{lemma}
\label{Lemma:DeletionandInsertion}
If for two binary strings of length $n$ condition (\ref{Eq:CodeRequirement}) is satisfied, then we have
\begin{align}
\label{Eq:InsertionConstraint}
D_{+1}\left( \mathbf{x} \right) \cap D_{+1}\left( \mathbf{y} \right) = \emptyset.
\end{align}
In other words, a single-deletion-error correcting code is a single-insertion-error correcting code.
\end{lemma}
\begin{proof}
Suppose two binary strings $\mathbf{x}$ and $\mathbf{y}$ of length $n$ satisfy (\ref{Eq:CodeRequirement}) but not (\ref{Eq:InsertionConstraint}). Suppose inserting $x^\prime$ in position $i$ of $\mathbf{x}$ and $y^\prime$ in position $j$ of $\mathbf{y}$ results in the same sequence of length $n+1$. Without loss of generality, we assume $i \leq j$. This implies that deleting $x_{j-1}$ and $y_{i}$ from $\mathbf{x}$ and $\mathbf{y}$ respectively, would result in the same sequence of length $n-1$. This contradicts (\ref{Eq:CodeRequirement}), thus proving the lemma.
\end{proof}

\vspace{1mm}
\noindent{{\bf Step~3: Our extended codeword can detect up to two shift errors.}} In this step, we show that using the six delimiter bits and the checksum condition of (\ref{eq:VTCondition}), we can detect up to two shift errors. Consider $m$ consecutive magnetic domains in racetrack memory. If no error happens, we need $m-1$ shifts to access and read the $m$ bits stored in these $m$ domains. We summarize the resulting observations in positions $m-5,m-4,\ldots,m$ for all possible error combinations in Table~\ref{Table:Syndrom}. In Appendix~\ref{Appendix:Table}, we explain how Table~\ref{Table:Syndrom} is obtained.

\begin{figure}[t]
\centering
\includegraphics[height = 6cm]{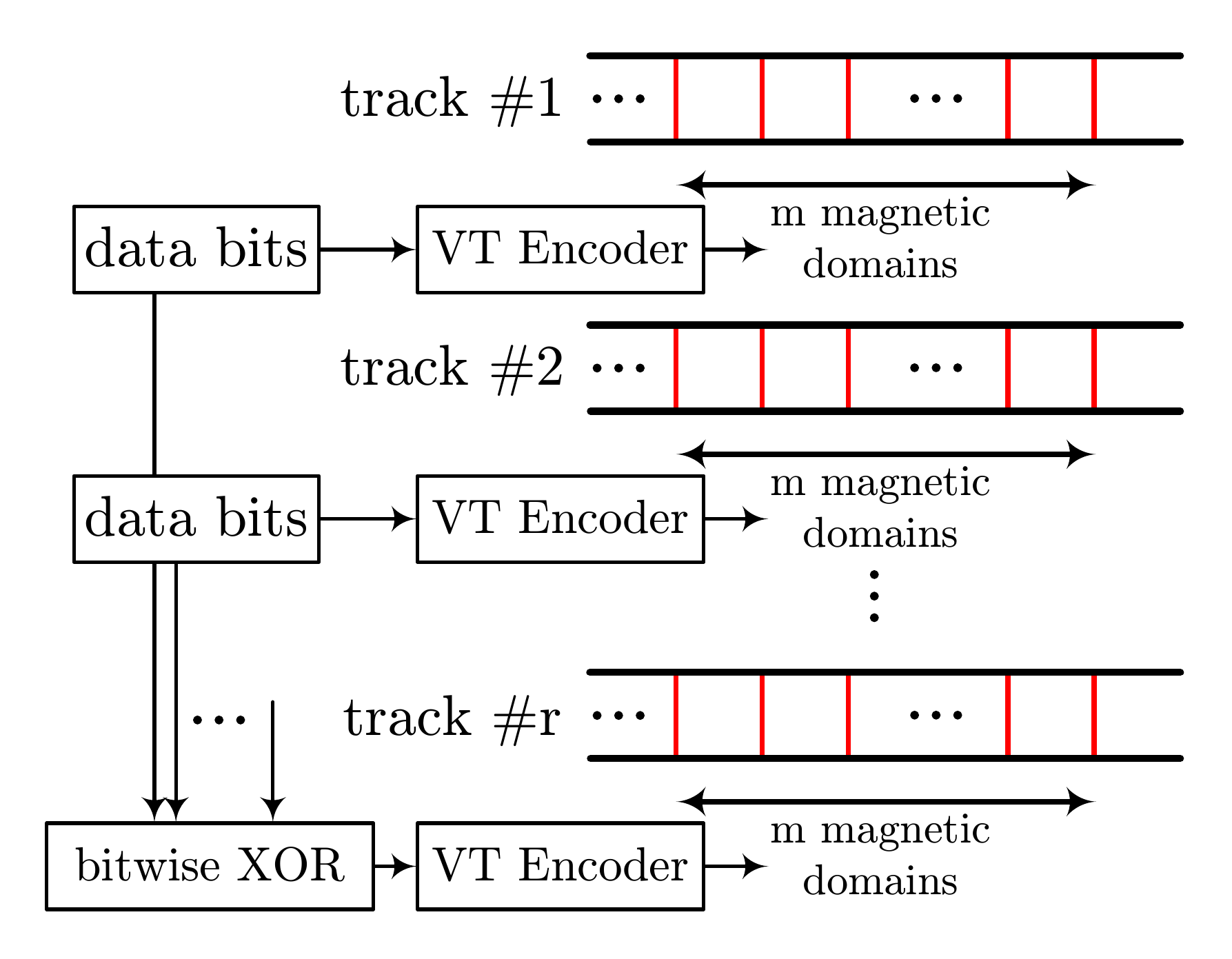}
\caption{Using an outer code on top of the code presented in Section~\ref{Section:Practical} allows us to correct two shift errors.\label{Fig:OuterCode}}
\vspace{-4mm}
\end{figure}

The most interesting case is when a combination of a deletion error and a repetition error takes place prior to the delimiter bits. In this case, we observe the delimiter bits as the error-free case, \emph{i.e.} $110000$. To detect such errors, we have the following result for VT codes.
\begin{lemma}
\label{Lemma:combination}
For any two binary strings  $\mathbf{x}, \mathbf{y} \in VT\left( n \right)$, $\mathbf{y}$ cannot be obtained from $\mathbf{x}$ by a deletion followed by an insertion (or by an insertion followed by a deletion).
\end{lemma}

\begin{proof}
Suppose $\mathbf{x}^\prime \in D_{-1}\left( \mathbf{x} \right)$ and $\mathbf{y}^\prime \in D_{-1}\left( \mathbf{y} \right)$. If $\mathbf{y}$ can be obtained from $\mathbf{x}$ by a deletion followed by an insertion, then 
\begin{align}
D_{+1}\left( \mathbf{x}^\prime \right) \cap D_{+1}\left( \mathbf{y}^\prime \right) \neq \emptyset.
\end{align}
In other words, $\mathbf{x}^\prime$ and $\mathbf{x}^\prime$ do not satisfy (\ref{Eq:InsertionConstraint}) which in turn implies that $\mathbf{x}$ and $\mathbf{y}$ do not satisfy (\ref{Eq:CodeRequirement}). However, this contradicts Lemma~\ref{Lemma:VTGraphBased}. This contradiction proves the result. A similar argument can be used when an insertion is followed by a deletion. 
\end{proof}
Lemma~\ref{Lemma:combination} immediately results in the following observation.
\begin{corollary}
A combination of a deletion error and an insertion error on $\mathbf{x} \in VT\left( n \right)$ either results in $\mathbf{x}$ or in $\mathbf{y} \notin VT\left( n \right)$ (an invalid codeword).
\end{corollary}

The shaded row in Table~\ref{Table:Syndrom} corresponds to a combination of a deletion error and an insertion error that resulted in an invalid codeword.

\noindent{{\bf Step~4: Our extended codeword can correct a single shift error and decode the data.}} First, we need to understand how to recover data in the error-free case.

\noindent {\bf Error-free decoding:} If no error is detected using the delimiter bits and the checksum, the $k = n-\ell-1$ data bits are simply retrieved from the positions in $\mathbf{c}$ that are not powers of two, \emph{i.e.} positions $3,5,6,7,9,$ etc. So the only remaining task is to correct a shift error when it happens. 

\vspace{1mm}
\noindent {\bf Correcting a deletion error:} Deletion errors can be corrected using the elegant algorithm proposed by Levenshtein~\cite{levenshtein1966binary,levenshtein1965binary}. This algorithm is presented below, and we refer the reader to~\cite{Sloane} for a mathematical analysis.
\begin{framed}
\begin{enumerate}[leftmargin=0.3cm]
\item Suppose a codeword $\mathbf{c} = \left( c_1, c_2, \ldots, c_n \right)$ from $VT\left( n \right)$ is stored in racetrack memory, and we detect a single deletion, and we observe $\mathbf{c}^\prime =  \left( c^\prime_1, c^\prime_2, \ldots, c^\prime_{n-1} \right)$.

\item Set $\omega$ to be the Hamming weight (number of $1$'s) of $\mathbf{c}^\prime$.

\item Calculate the checksum $\sum_{i=1}^{n-1}ic^\prime_i$ and set $s$ to be the minimum amount that needs to be added to the checksum in order to make it $0~\textbf{mod}~n+1$.

\item If $s \leq \omega$, we restore the codeword by adding a $0$ immediately to the left of the rightmost $s$ $1$'s. Otherwise, we restore the codeword by adding a $1$ immediately to the right of the leftmost $\left( s-\omega-1 \right)$ $0$'s. 
\end{enumerate}
\end{framed}

Using this algorithm, we reconstruct the correct codeword and recover the dataword as discussed for the error free case. Note that while the decoding algorithm corrects the codeword, it does not reveal where the deletion occurred. 

\vspace{1mm}
\noindent {\bf Correcting an insertion error:} From Lemma~\ref{Lemma:DeletionandInsertion}, we know that VT codes can correct an insertion. The corresponding Insertion Correction Algorithm is a modified version of the Deletion Correction Algorithm presented above and is omitted here due to space limitations.
%
%
%

\vspace{1mm}
\noindent{{\bf Step~5: The outer code corrects two shift errors.}} Suppose we detect two shift errors and the data is erased as in Table~\ref{Table:Syndrom}. We can use an outer code to recover data in this case as illustrated in Fig.~\ref{Fig:OuterCode} where the outer code is a simple single parity code. We note that per track we can correct single shift errors and detect up to two shift errors.

%



\appendices

\section{Proof of Lemma~\ref{Lemma:VTGraphBased}}
\label{Appendix:AlgoDescription}

In this appendix, we prove Lemma~\ref{Lemma:VTGraphBased}. To prove the result, it suffices to show that after a single deletion, Varshamov-Tenengolts codes cannot be confused and unique recovery is guaranteed. We borrow from the argument given in~\cite{Sloane}. 

Suppose a codeword $\mathbf{c} = \left( c_1, c_2, \ldots, c_n \right)$ from $VT\left( n \right)$ is stored in Racetrack memory; and bit $x \in \{ 0, 1 \}$ in position $1 \leq p \leq n$ is deleted to observe $\mathbf{c}^\prime = \left( c^\prime_1, c^\prime_2, \ldots, c^\prime_{n-1} \right)$. Let there be $L_0$ $0$'s and $L_1$ $1$'s to the left of $x$, and $R_0$ $0$'s and $R_1$ $1$'s to the right of $x$ (with $p = 1 + L_0 + L_1$). Also, let $\omega$ denote the weight (number of $1$'s) of $\mathbf{c}^\prime$. Calculate the new checksum $\sum_{i=1}^{n-1}ic^\prime_i$.

Now if $x = 0$, the new checksum is $R_1$ less than $\sum_{i=1}^{n}ic_i$. We note that $R_1 \leq \omega$. If $x = 1$, the new checksum is $p+R_1$ less than $\sum_{i=1}^{n}ic_i$. We note that $p+R_1 = 1 + L_0 + L_1 + R_1 = 1 + \omega + L_0$ which is greater than $\omega$.

So if the deficiency in the checksum is less than or equal to $\omega$, we restore the codeword by adding a $0$ just to the left of the rightmost $R_1$ $1$'s. Otherwise, we restore the codeword by adding a $1$ just to the right of the leftmost $L_0$ $0$'s.

\section{Derivation of Table~\ref{Table:Syndrom}}
\label{Appendix:Table}

In this appendix we describe how the combination of a VT code and six delimiter bits can be used to to identify two deletions, or a deletion and a repetition. The other cases appearing in Table~\ref{Table:Syndrom} can be derived using similar arguments.

\noindent {\bf Two deletions:} There are three cases to consider as listed below.
\begin{enumerate}

\item {\bf Two deletions in the VT code:} In  this case, in positions $m-5, m-4, m-3,$ and  $m-2$ we observe $0000$ and position $m-1$ corresponds to the first bit of the next extended codeword. In this case, we declare erasure and we use the outer-code to resolve this case.

\item {\bf Two deletions in the delimiter bits:} For this scenario, the VT code is unaffected. However, we will declare either one or two deletions as we cannot identify the location of these errors. In positions $m-5, m-4, m-3,$ and  $m-2$ we can observe three possibilities as described below and position $m-1$ corresponds to the first bit of the next extended codeword.
\begin{enumerate}

\item $0000$: Similar to the case when two deletions happen in the VT code, we declare an erasure. The VT code is not affected but since we cannot determine that the errors were within the delimiter bits, we cannot trust the checksum. We declare erasure and we use the outer-code to resolve this case.

\item $1000$: According to Table~\ref{Table:Syndrom}, we declare a single deletion. Again note that the VT code is not affected here. So even if we assume a single deletion, data can be reliably recovered. The second deletion error will be detected when we reach the next set of delimiter bits are obtained. We note that our assumption of at most two shift errors every $m+3$ shifts is needed to guarantee that the second shift error will be detected.

\item $1100$: According to Table~\ref{Table:Syndrom}, this is an error free case. Since the VT code is not affected here, data can be reliably recovered. However, if there was no error at all, positions $m-1$ and $m$ would correspond to the last two of the six delimiter bits. On the other hand, here positions $m-1$ and $m$ correspond to the first two bits of the next extended codeword and we deal with that when the next set of delimiter bits are obtained.

\end{enumerate} 

\item {\bf One deletion in the VT code and one deletion in the delimiter bits:} For this case, in positions $m-5, m-4, m-3,$ and  $m-2$ we can observe two possibilities as described below and position $m-1$ corresponds to the first bit of the next extended codeword.
\begin{enumerate}

\item $0000$: Similar to the case when two deletions happen in the VT code, we declare an erasure. The VT code only contains a single deletion but since we cannot determine this fact, we cannot trust the checksum. We use the outer-code to resolve this case.

\item $1000$: According to Table~\ref{Table:Syndrom}, we declare a single deletion. Since the VT code only contains a single shift error, data can be reliably recovered. 

\end{enumerate}

\end{enumerate}

\noindent {\bf One deletion and one repetition:} We discussed the case when one deletion and one repetition happen within the VT code in Step~3 of Section~\ref{Section:Practical}. Suppose one deletion happens in the VT code and one repetition in the delimiter bits. For this case, in positions $m-5, m-4, m-3,$ and  $m-2$ we can observe two possibilities as described below. The case in which one repetition happens in the VT code and one deletion in the delimiter bits can be explained similarly.
\begin{enumerate}

\item $1000$: According to Table~\ref{Table:Syndrom}, we declare a single deletion. Since the VT code only contains a single shift error, data can be reliably recovered. 

\item $1100$: If the checksum is zero modulo $n+1$, we were lucky and the deletion followed by introducing the first bit of the delimiter bits in position $m-6$ resulted in the correct codeword. If the checksum is not zero modulo $n+1$, we declare erasure and we use the outer-code to resolve this case.
\end{enumerate}

\bibliographystyle{ieeetr}
\bibliography{bib_Racetrack}

\end{document}